\documentclass[a4paper]{article}

\usepackage[disable]{todonotes}

\usepackage[utf8]{inputenc}
\usepackage{xcolor}

\usepackage{fullpage}
\usepackage{parskip}
\usepackage{amsmath}
\usepackage{amsthm}
\usepackage{mathtools}
\usepackage{amssymb}
\usepackage{tabularx}
\usepackage[toc,page]{appendix}
\usepackage{xspace}
\usepackage{comment}

\usepackage{microtype}
\usepackage{bbm}
\usepackage{amsthm}
\usepackage{enumerate}
\usepackage[]{hyperref}

\usepackage{comment}
\usepackage[classfont=sanserif,langfont=sanserif,funcfont=sanserif]{complexity}
\usepackage{xcolor}

\newtheorem{theorem}{Theorem}[section]
\newtheorem{corollary}[theorem]{Corollary}

\newtheorem{lemma}[theorem]{Lemma}

\newtheorem{conjecture}[theorem]{Conjecture}

\newtheorem{definition}[theorem]{Definition}

\newcommand{\ket}[1]{|{#1}\rangle}

\newcommand{\RAG}{\mathsf{RAG}}
\newcommand{\CNOT}{\mathsf{CNOT}}

\usepackage{mathtools}

\newcommand{\ZO}{\{0,1\}}
\newcommand{\tO}{\widetilde O}

\newcommand{\matchingTriangles}{\textsc{$\Delta$-Matching Triangles}}
\newcommand{\triangleCollection}{\textsc{Triangle Collection}}

\newcommand{\ksat}{\textup{\textsc{$k$-SAT}}}
\newcommand{\cnfsat}{\textup{\textsc{CNF-SAT}}}
\newcommand{\sat}{\textup{\textsc{SAT}}}
\newcommand{\tsum}{\textup{\textsc{3SUM}}}
\newcommand{\apsp}{\textup{\textsc{APSP}}}
\newcommand{\dmt}{\textup{\textsc{$\Delta$-Matching Triangles}}}
\newcommand{\tc}{\textup{\textsc{Triangle Collection}}}
\newcommand{\zwt}{\textup{\textsc{$0$-Weight Triangle}}}
\newcommand{\mm}{\textup{\textsc{$(\min,+)$-Matrix Multiplication}}}
\newcommand{\apnt}{\textup{\textsc{All-Pairs Negative Triangle}}}
\newcommand{\nt}{\textup{\textsc{Negative Triangle}}}
\newcommand{\conv}{\textup{\textsc{Convolution-3Sum}}}

\usepackage{pgf}
\usepackage{tikz}
\usepackage[utf8]{inputenc}
\usetikzlibrary{arrows,automata}
\usetikzlibrary{positioning}

\tikzset{
    state/.style={
           rectangle,
           rounded corners,
           draw=black, very thick,
           minimum height=2em,
           inner sep=2pt,
           text centered,
           },
}
\usepackage{cleveref}

\usepackage{tabularx}
\usepackage{makecell,multirow}

\title{Matching Triangles and Triangle Collection: \\Hardness based on a Weak Quantum Conjecture}

\author{Andris Ambainis\thanks{Center for Quantum Computer Science, Faculty of Computing, University of Latvia, Latvia. {\tt andris.ambainis@lu.lv}}
\and
Harry Buhrman\thanks{QuSoft, CWI and University of Amsterdam. {\tt harry.buhrman@cwi.nl}}
\and 
Koen Leijnse\thanks{QuSoft, University of Amsterdam. {\tt koenleijnse@gmail.com}}
\and 
Subhasree Patro\thanks{QuSoft, CWI and University of Amsterdam.  {\tt Subhasree.Patro@cwi.nl}}
\and 
Florian Speelman\thanks{QuSoft, CWI and University of Amsterdam. {\tt f.speelman@uva.nl}}}

\date{\today}

\begin{document}

\maketitle

\begin{abstract}
Classically, for many computational problems one can conclude time lower bounds conditioned on the hardness of one or more of key problems, such as $\ksat$, $\tsum$ and $\apsp$. More recently, similar results have been derived in the quantum setting conditioned on the quantum hardness of $\ksat$ and $\tsum$
\cite{aaronson_quantum_2020,buhrman_limits_2021,buhrman_framework_2021}.
This is done using fine-grained reductions, where the approach is to (1) select a key problem $X$ that for some function $T$, is conjectured to not be solvable by any $O(T(n)^{1-\epsilon})$ time algorithm for any constant $\epsilon > 0$ (in a fixed model of computation), and (2) reduce $X$ in a fine-grained way to these computational problems, thus giving (mostly) tight conditional time lower bounds for them. 

Interestingly, for $\matchingTriangles$ and $\triangleCollection$, two graph problems with natural definitions, classical hardness results have been derived~\cite{abboud_matching_2018} conditioned on hardness of all three key problems. More precisely, it is proven that an $O(n^{3-\epsilon})$ time classical algorithm (for any constant $\epsilon > 0$) for either of these two graph problems would imply faster classical algorithms for $\ksat$, $\tsum$ and $\apsp$, which makes $\matchingTriangles$ and $\triangleCollection$ worthwhile to study. 

In this paper, analogous to the classical result, we show that an $O(n^{1.5-\epsilon})$ time quantum algorithm (for any constant $\epsilon > 0$) for either of these two graph problems would imply faster quantum algorithms for $\ksat$, $\tsum$ and $\apsp$. To prove these results, we first formulate a quantum hardness conjecture for $\apsp$ analogous to the classical one and then present quantum reductions from $\ksat$, $\tsum$ and $\apsp$ to $\matchingTriangles$ and $\triangleCollection$. Additionally, based on the quantum $\apsp$ conjecture, we are also able to prove quantum lower bounds for a matrix problem and many graph problems. The matching upper bounds follow trivially for most of them, except for $\matchingTriangles$ and $\triangleCollection$ for which we present quantum algorithms that require careful use of data structures and Ambainis' variable time search \cite{ambainis_quantum_2006} as a subroutine.
\end{abstract}

\thispagestyle{empty}

\tableofcontents{}

\newpage
\clearpage
\setcounter{page}{1}

\section{Introduction}
\label{sec:Introduction}

Recent advancements in quantum hardware technologies have made it even more exciting to further the theoretical development of quantum algorithms, because of the possible speedups for computational problems as compared to their respective classical counterparts. Interestingly, for some naturally occurring problems, we can prove limits on how much quantum speedup is achievable. For example, it can be shown that it is not possible to get a super-quadratic quantum speedup for the \emph{unordered search} problem. Most of these results follow from quantum \emph{query} lower bounds that don't always immediately imply optimal \emph{time} lower bounds, especially for problems that require super-linear time. In general, for a powerful enough computational model, time lower bounds are notoriously hard to obtain. To overcome this barrier, several recent works have been developing the field of quantum fine-grained complexity, which has the power to prove conditional quantum time lower bounds for many computational problems and thereby elucidates the internal structure of the BQP complexity class.

The idea of fine-grained complexity is as follows:
Pick a well-studied problem $X$, that on $n$ input variables is conjectured to not be solvable in $O(T(n))^{1-\epsilon}$ time in a particular model of computation, for any constant $\epsilon > 0$, for some monotonically increasing function $T$.
Then, given an input instance of $X$, if we can reduce it to input instances of some problem $Y$ such that we can solve $X$ using any algorithm for $Y$, we have found a time lower bound for $Y$ based on the conjectured hardness of $X$ in the same model of computation.
Classically, some of the most well studied problems are $\sat$, $\tsum$ and $\apsp$, and, based on the hardness of these key problems classical time lower bounds for a lot of computational problems have been derived.\footnote{Also see the survey article by \cite{williams_fine-grained_2019} for a summary of these results and their corresponding reductions.}
Note that this logic gives us a double-edged sword, in a positive way:
If a speed-up is possible for any of the key problems (which happens if there's a speed-up for any problem they are reduced to) we are happy with having a new result about a fundamental problem.
On the other hand, if no speed-up is possible, then the web of fine-grained reductions imply interesting lower bounds for a swathe of natural problems.

In this paper, we further extend the field of fine-grained complexity by presenting similar results in the quantum setting. We give a quantum hardness conjecture for the All Pairs Shortest Path problem ($\apsp$), a classical key problem, and present structural fine-grained results that can be derived from this.
In particular, we show fine-grained quantum lower bounds for two graph problems conditional on several key problems at the same time (and a matching non-trivial upper bound), culminating in a quantum analogue of the result of Abboud, Williams, and Yu~\cite{abboud_matching_2018}.

\subsection{A Weak Quantum Conjecture}
\label{sec:PopularConjecture}

One of the first problems to be studied in this context of quantum fine-grained complexity is $\ksat$, the problem of whether a formula, input in conjunctive normal form, has a satisfying assignment. Classically, $\ksat$ is conjectured to require $2^n$ time, referred to as the Strong Exponential-Time Hypothesis (SETH). However, SETH fails to hold in the quantum setting because one can solve $\ksat$ in $O^*(2^{n/2})$ time. Recently, two independent works studied quantum variants of SETH, and used these variants to prove (often tight) bounds on how much quantum speed-up is possible for various problems. The quadratically faster quantum algorithm to solve $\ksat$ is the current best known algorithm, and no further significant improvement to the run-time is known. Consequently, this led to the following conjecture.

\begin{conjecture}[QSETH \cite{aaronson_quantum_2020}]
\label{conj:basicQSETH}
For every $\delta > 0$ there exists a $k \geq 3$ such that there is no bounded-error quantum algorithm that solves $\ksat$ on $n$ variables, $m$ clauses in $O(m^{O(1)}2^{\frac{n}{2}(1-\delta)})$ time.
\end{conjecture}

Aaronson, Chia,
Lin, Wang, and Zhang \cite{aaronson_quantum_2020} presented linear quantum time lower bounds for Closest Pair, Bichromatic Closest Pair, and Orthogonal Vectors, using QSETH. In the same paper, they also present matching quantum upper bounds for these problems. Simultaneously, Buhrman, Patro, and Speelman \cite{buhrman_framework_2021} presented a framework for proving quantum time lower bounds for many problems in $\P$ conditioned on quantum hardness of variants of $\sat$, which they used to prove an $n^{1.5}$ quantum time lower bound for the Edit Distance and the Longest Common Subsequence problems. 

This kind of quantum fine-grained results was next shown to include hardness results based on the quantum hardness of $\tsum$, a problem when given a list of $n$ integers, one needs to output if the list contains a triple $a,b,c$ such that $a+b+c=0$. While the classical complexity of $\tsum$ is $\Theta(n^2)$ with a conjectured lower bound, it turns out that there is a $\tO(n)$ time quantum algorithm to solve $\tsum$. However, because of lack of significant improvement to the run-time, the following was conjectured.

\begin{conjecture}[Quantum $\tsum$ conjecture]
\label{conj:Q3SUM}
There is no bounded-error quantum algorithm that solves $\tsum$ on a list of $n$ integers in $O(n^{1-\delta})$ time, for any constant $\delta >0$.
\end{conjecture}

Based on this conjecture (almost) tight quantum bounds were derived for $\conv$, $\zwt$ and a long list of computational geometry problems \cite{buhrman_limits_2021}.

In this paper, we further extend the field of fine-grained complexity by studying the quantum hardness of the All Pairs Shortest Path ($\apsp$) problem which is defined as follows: Given a weighted graph $G=(V, E)$ on $n$ nodes with no negative cycles, for every pair of vertices $(a,b) \in V \times V$, output the shortest distance between vertices $a$ and $b$ if there is a path, else output $\infty$.

The currently fastest known classical algorithm for $\apsp$ runs in $n^{3}/\exp({\sqrt{\log n}})$ time \cite{williams_faster_2018}, and it has been conjectured that no $O(n^{3-\epsilon})$ time classical algorithm, for any constant $\epsilon >0$, is possible. Based on this conjecture, time lower bounds for a lot of problems, for example Graph Radius, Graph Median, Negative Triangle and many more, have been concluded. However, the fastest known quantum algorithm to solve to solve $\apsp$ takes time $\tO(n^{2.5})$  \cite{durr_quantum_2006}, but no significant speedups to this algorithm have been found in a long time.\footnote{Also see Koen Leijnse's master's thesis for a discussion on this \cite{thesis_koen}.} Therefore, it is natural to study the consequences of the following conjecture. 

\begin{conjecture}[Quantum APSP Conjecture]
\label{conj:QAPSP}
There is no bounded-error quantum algorithm that solves $\apsp$ on a graph of $n$ nodes in $O(n^{2.5-\delta})$ time, for any constant $\delta >0$.
\end{conjecture}

We study the classical reductions from $\apsp$ to six problems and observe that almost all those reductions can be trivially adapted to the quantum setting, thus proving tight lower bounds for nearly all of these problems. The matching upper bounds for most of them can be derived using Grover-like speed-ups, except for $\matchingTriangles$ and $\triangleCollection$ for which we present quantum algorithms that require careful use of data structures and Ambainis' variable time search as a subroutine, which we discuss shortly after.

In the classical case, hardness results for $\matchingTriangles$ and $\triangleCollection$ can be derived conditioned on hardness of all the three key problems. More precisely, Abboud, Williams, and Yu have proven that an $O(n^{3-\epsilon})$ time classical algorithm (for any constant $\epsilon > 0$ and $\omega(1) \leq \Delta(n) \leq n^{o(1)}$) for either of these two graph problems would imply faster classical algorithms for $\ksat$, $\tsum$ and $\apsp$ \cite{abboud_matching_2018}, which makes $\matchingTriangles$ and $\triangleCollection$ worthwhile to study.
Which means, one can now make fewer hardness assumptions when it comes to understanding hardness of $\matchingTriangles$ and $\triangleCollection$.  Towards that, a weaker conjecture was introduced  which states that at least one of classical $\tsum$-conjecture, $\apsp$-conjecture or SETH is true. They called it the \emph{extremely popular conjecture}.

Analogous to Abboud et al.'s classical result, we are also able to show that an $O(n^{1.5-\epsilon})$ time quantum algorithm (for any constant $\epsilon > 0$ and the same ranges of $\Delta$) for either of these two graph problems would imply faster quantum algorithms for $\ksat$, $\tsum$ and $\apsp$. Clearly, for problems such as $\matchingTriangles$ and $\triangleCollection$ a weaker quantum hardness assumption can be made to conclude time lower bounds. Hence, we state the following conjecture analogous to the classical case. 

\begin{conjecture}
\label{conj:PopularConjecture}
At least one of Conjecture~\ref{conj:basicQSETH},~\ref{conj:Q3SUM}~or~\ref{conj:QAPSP} is true.
\end{conjecture}

 Based on Conjecture~\ref{conj:PopularConjecture}, we are able to prove tight quantum time bounds for $\matchingTriangles$ and $\triangleCollection$. The upper bounds are non-trivial and will be the content of the next sub-section.

\subsection{Quantum Upper Bounds for $\matchingTriangles$ and $\triangleCollection$}
The $\matchingTriangles$ and $\triangleCollection$ problems are two naturally occurring graph problems, whose definitions are as follows.

\begin{definition}[$\matchingTriangles$] Given a graph $G=(V,E)$ with a colouring of the vertices $\gamma: V \rightarrow \Gamma$ with $|\Gamma| \leq n$, determine if there is a triple of colours $i,j,k \in \Gamma$ such that there are at least $\Delta$ triangles $a,b,c \in V$ for which $(\gamma(a),\gamma(b),\gamma(c))=(i,j,k)$.
\end{definition}

Note that the range of $\Delta$ can vary between $0 \leq \Delta \leq n^3$.

\begin{definition}[$\triangleCollection$]Given a graph $G=(V,E)$ with a colouring of the vertices $\gamma: V \rightarrow \Gamma$ with $|\Gamma| \leq n$, determine if for every triple of colours $i,j,k \in \Gamma$ there is at least one triangle $a,b,c \in V$ for which $(\gamma(a),\gamma(b),\gamma(c))=(i,j,k)$.
\end{definition}

Quantum algorithms for many of the graph problems we consider in this work, such as $\nt$ or $\zwt$, can be found using an easy application of Grover's algorithm.
However, the algorithms for $\matchingTriangles$ and $\triangleCollection$ are less trivial -- these will heavily use \emph{Variable Time Search} by Ambainis~\cite{ambainis_quantum_2006} which can be informally stated as follows.

\begin{theorem}[Informal]
\label{thm:VariableTimeSearch}
Given a string $x \in \{0,1\}^n$, the task is to find $i$ such that $x_i=1$. Additionally, let $t_i$ be the maximum number of time steps required to evaluate $x_i$. Then the total time steps to find an $i$ such that $x_i=1$ is
\begin{equation*}
    O(\sqrt{t_1^2+t_2^2+ \ldots +t_n^2}).
\end{equation*}
\end{theorem}

The algorithm for $\matchingTriangles$ is more complicated the one for $\triangleCollection$, with the similarity that this also makes use of the Variable Time Search of Theorem~\ref{thm:VariableTimeSearch}. Hence, we will sketch the algorithm for $\triangleCollection$ first and proceed to discuss the intuition of the algorithm for $\matchingTriangles$ after that. Also see \Cref{c5} for the details of these algorithms.

\paragraph{Algorithm for $\triangleCollection$} Recall the definition of the $\triangleCollection$ problem: Given a graph $G=(V,E)$, we want to know if for every triple of colours there is a triangle in the graph $G=(V,E)$. This is the same as knowing if there is a triple of colours such that there is \emph{no} triangle of that colour triple in $G$. Making use of this simple observation we do the following: Let's assume a subroutine that given a colour triple $(i,j,k) \in \Gamma^3$ outputs yes if there is a triangle of this colour in the input graph $G$, and no otherwise. Furthermore, let $t_{i,j,k}$ be the time taken for the subroutine on colour triple $(i,j,k)$. Invoking Theorem~\ref{thm:VariableTimeSearch} we can now conclude that the total time taken on a graph of $n$ nodes is $T(n)=O(\sqrt{\sum_{(i,j,k) \in \Gamma^3}t_{i,j,k}^2})$. Suppose that with some pre-processing of the input, this is where the data structures come to use, we could \emph{efficiently} access nodes of $G$ coloured by $i$ for any $i \in \Gamma$. Then in $t_{i,j,k}=\tO(\sqrt{|V_i| \cdot |V_j| \cdot |V_k|})$ time, where $V_i, V_j, V_k$ denotes the sets of nodes with colours $i,j,k$ respectively, we can find if there is a triangle in $G$ of colour $i,j,k$. Which means
\begin{equation*}
    T(n)=O(\sqrt{\sum_{i,j,k \in \Gamma^3}t_{i,j,k}^2})=\tO(\sqrt{\sum_{i,j,k \in \Gamma^3} |V_i|\cdot|V_j|\cdot|V_k|})=\tO(\sqrt{\sum_{i \in \Gamma} |V_i| \sum_{j \in \Gamma}|V_j|\sum_{k \in \Gamma}|V_k|})=\tO(n^{1.5}).
\end{equation*}

On the other hand, the algorithm for $\matchingTriangles$ is slightly more complicated. Apart from using the \emph{Variable Time Grover Search} subroutine, the algorithm makes use of different ways of counting number of triangles in a graph. To state some of these several methods:
\begin{enumerate}
    \item Given a graph $G=(V,E)$ of $n$ nodes and access to the adjacency matrix $A_G$ corresponding to $G$, one can count the number of triangles by computing the trace of $A_G^3$. Here $A_G^3$ refers to matrix multiplication of $A_G$ with itself three times. This can be achieved classically in $O(n^{\omega})$ time, where $\omega$ denotes matrix multiplication constant currently known to be at $2.3728$.
    \item Or, one could use the threshold variant of Grover Search, with which one can check if there are at least $k$ triangles in a graph of $n$ nodes in $O(n^{1.5+\frac{k}{2}})$ time.
\end{enumerate}

\paragraph{Intuition for Algorithm for $\matchingTriangles$} Recall the definition of $\matchingTriangles$: Given a graph $G=(V,E)$ with $n$ nodes and a colouring of the nodes $\gamma: V \rightarrow \Gamma$, is there a triple of colour $(i,j,k) \in \Gamma^3$ such that there are at least $\Delta$ triangles of that colour triple. Clearly, trivial brute forcing (in spite of the Grover like speedup) over all triples of colours and checking if any of the colour triple has at least $\Delta$ triangles is going to cost $O(\sqrt{\Delta n^3})$. While this bound is sub-cubic for small $\Delta$, however when $\Delta \approx n^3$ we get the same upper bound as the classical one. What helps is firstly the observation that when $\Delta$ is \emph{too big} then there aren't that many triples of colours that one needs to brute force over. Secondly, having fixed a triple of colour, it is more efficient to use brute force search over matrix multiplication if the number of nodes restricted to the triple of colour is small. Using these two simple but crucial observations, we get $\tO(n^{1.5}\sqrt{\Delta})$ for $1\leq \Delta(n) \leq n^{\omega}$ and $\tO(\frac{n^{1.5 + \omega}}{\sqrt{\Delta}})$ for $n^{\omega}\leq \Delta(n) \leq n^3 $. We present the details and the exact parameters of the algorithm in \Cref{c5}.

To summarise, in this paper we continue the recent work done in quantum fine-grained complexity: first by formulating a quantum hardness conjecture for $\apsp$, then by studying reductions from $\apsp$ to several computational problems. See \Cref{table} for an overview of the results. Furthermore, for $\matchingTriangles$ and $\triangleCollection$, we present quantum reductions from all the three key problems, i.e., $\ksat$, $\tsum$ and $\apsp$, from which we are able to derive conditional lower bounds for $\matchingTriangles$ and $\triangleCollection$ based on Conjecture~$\ref{conj:PopularConjecture}$, which (by construction) is a weaker conjecture than any of the earlier quantum hardness conjectures. We also provide non-trivial quantum upper bounds for the same. See \Cref{fig:flowchart} for the overview of how these problems computationally relate to each other.

\begin{table}[t]
\begin{tabular}{|l|l|l|l|}
\hline
                 Problem &   & Classical & Quantum \\ \hline &&& \\[-1em]
\multirowcell{2}{\textsc{$(\min,+)$-Matrix} \\ \textsc{Multiplication}} & Lower bound & $n^{3-o(1})$ \cite{fischer_boolean_1971,trans_closure} & $n^{2.5-o(1)}$ \Cref{QAPSPtoMM} \\ \cline{2-4}&&&\\[-1em] 
                  & Upper bound & $O(n^{3})$ $(\ast)$& $O(n^{2.5})$ ($\ast\ast$)\\ \hline &&&
\\[-1em]
\multirowcell{2}{\textsc{All-Pairs} \\  \textsc{Negative Triangle}} & Lower bound & $n^{3-o(1)}$ \cite{williams_subcubic_2018} & $n^{2.5-o(1)}$ \Cref{QMMtoAPNT} \\ \cline{2-4}&&& \\[-1em]
                  & Upper bound & $O(n^{3})$ $(\ast)$ & $O(n^{2.5})$ $(\ast\ast)$
                 \\ \hline&&&\\[-1em] 
\multirow{2}{*}{\textsc{Negative Triangle}} & Lower bound & $n^{3-o(1)}$  \cite{williams_subcubic_2018}& $n^{1.5-o(1)}$ \Cref{QAPNTtoNT} \\ \cline{2-4}&&& \\[-1em]
                  & Upper bound & $O(n^{3})$ $(\ast)$& $O(n^{1.5})$ $(\ast\ast)$\\ \hline&&&\\[-1em]
\multirow{2}{*}{\textsc{0-Weight Triangle}} & Lower bound& $n^{3-o(1)}$  \cite{vassilevska_finding_2009}& $n^{1.5-o(1)}$ \Cref{QNTto0WT}\\ \cline{2-4}&&& \\[-1em]
                  & Upper bound & $O(n^{3})$ $(\ast)$& $O(n^{1.5})$ $(\ast\ast)$ \\ \hline&&&\\[-1em]
\multirowcell{2}{\textsc{$\Delta$-Matching} \\ \textsc{Triangles}} & Lower bound & $n^{3-o(1)}$ \cite{abboud_matching_2018} $(\ast\ast\ast)$& $n^{1.5-o(1)}$ \Cref{Q0WTtoDMT} $(\ast\ast\ast)$\\ \cline{2-4}&&& \\[-1em]
                  & Upper bound & $O(n^{3-o(1)})$ \cite{abboud_matching_2018} $(\ast\ast\ast)$& $\tO(n^{1.5+o(1)})$  \Cref{qDMTsmallD} $(\ast\ast\ast)$\\ \hline&&&\\[-1em]
\multirowcell{2}{\textsc{Triangle} \\ \textsc{Collection}} & Lower bound & $n^{3-o(1)}$ \cite{abboud_matching_2018} & $n^{1.5-o(1)}$ \Cref{Q0WTtoTC} \\ \cline{2-4}&&&\\[-1em] 
                  & Upper bound & $O(n^{3})$ $(\ast)$& $\tO(n^{1.5})$ \Cref{QTCu}\\ \hline
\end{tabular}
\caption{Overview of lower bounds based on a hardness conjecture for $\apsp$, both in the classical and in the quantum setting. Corresponding upper bounds are also provided.
\newline
$(\ast)$: These upper bounds are the most straight forward algorithms, like exhaustive search, and therefore have no particular source. \newline
$(\ast\ast)$: By applying Grover Search, potentially as subroutine.
\newline
$(\ast\ast\ast)$: Holds only for $\omega(1) \leq \Delta \leq n^{o(1)}$.}
\label{table}
\end{table}

\subsection{Related Work}
Recently the field of quantum fine-grained complexity was also expanded to introduce a quantum analogue of the HSC, stating that the Hitting-Set (HS) problem has at most a quadratic speed-up in the quantum setting -- for this,  see Daan Schoneveld's Bachelor's Thesis~\cite{thesis_daan}. Using this conjecture, it was possible to prove a linear quantum lower bound for the  $(3/2 - \varepsilon)$-approximate radius problem. The same lower bound is also proven for the  $(3/2 - \varepsilon)$-approximate diameter problem, using the quantum version of SETH.

\subsection{Structure of the Paper}
We briefly state the computational model in \Cref{sec:Preliminaries}. In \Cref{c4}, we present quantum fine-grained reductions from three key problems, $\ksat, \tsum$ and $\apsp$, to several computational problems. Additionally, for $\matchingTriangles$ and $\triangleCollection$, we give reductions from all the three key problems. In \Cref{c5}, we give non-trivial quantum time upper bounds for $\matchingTriangles$ and $\triangleCollection$. Finally, we conclude in \Cref{c6}.

\section{Preliminaries}
\label{sec:Preliminaries}
We let $[n]=\{1, \ldots, n\}$. The $\tO(\cdot)$ notation hides poly-logarithmic factors in $n$.

\textbf{Quantum Computational Model.} For our algorithms, we consider the standard quantum circuit model of computation with the gate set containing all single qubit unitaries, $\CNOT$ and Random-Access Gates ($\RAG$s) which is defined as follows
\begin{equation}
\label{eq:RAG}
    \RAG \ket{i, b, x_0, \ldots, x_{n-1}} = \ket{i, x_i, x_0, \ldots, x_{i-1}, b, x_{i+1}, \ldots, x_{n-1}} \qquad \forall i\in[n], b, x_0, \ldots, x_{n-1} \in \ZO.
\end{equation}
The $\RAG$s are necessary to our algorithms and also in some reductions because without them simple operations that require $O(1)$ time would require $O(n)$ single and double qubit unitaries in the quantum setting. For further details refer to \cite{buhrman_memory_2022} where the model of computation is formally defined. 
In this computational model, we have can access the input (in superposition) using a bit oracle $O_x$, which is defined as \footnote{Given access to a bit oracle, we can convert it into a phase oracle, i.e., $O_{x,\pm}\ket{i} = (-1)^{x_i}\ket{i}$.} $O_x \ket{i,b} = \ket{i,b\oplus x_i}$.

Let $C_n$ denote the circuit corresponding to a quantum algorithm on input length $n$. We define time complexity of the algorithm (or equivalently of the circuit) to be the total number of gates in $C_n$. Single applications of $\RAG$ and $O_x$ are counted as one gate, respectively, even though in practice it is much costlier to implement these gates using single and two qubit unitaries.

\section{Quantum Fine-Grained Conditional Lower Bounds}
\label{c4}
In this section we will review quantum fine-grained reductions from, $\ksat$, $\tsum$ and $\apsp$ to $\dmt$ and $\tc$ to prove the hardness of these two triangle problems based on Conjecture~\ref{conj:PopularConjecture}, a quantum version of the hardness results based on extremely popular conjecture from \cite{abboud_matching_2018}. We also show a series of conditional quantum lower bounds for different computational problems in the process. 

\subsection{Quantum Fine-Grained Reductions from $\apsp$}
The quantum lower bounds that we show here use the same reductions that were used for proving the classical lower bounds in different works: \cite{fischer_boolean_1971,trans_closure,abboud_matching_2018,williams_subcubic_2018,vassilevska_finding_2009}. We will therefore not go into much detail here for proving the correctness of the reductions. Some adjustments are necessary however, for boosting the accuracy of our subroutines. We will therefore sketch the reductions briefly.\footnote{A more comprehensive overview of all the classical reductions can be found in \cite{thesis_koen}.}

The first couple of results are quite straight-forward, we can use the classical reductions with occasional on-the-fly methods, and we find many quantum lower bounds based on the quantum hardness of $\apsp$. 
First we retain the hardness-equivalence with~$\mm$ that was first discovered in \cite{fischer_boolean_1971,trans_closure}. These two reductions are somewhat considered folklore, for detailed descriptions we refer to \cite{thesis_koen}.

\begin{definition}[$\mm$]
Given two $n\times n$ matrices $M$, $N$ with entries from $[-n^c,n^c]$, compute the distance product $M \star N$ defined as: $(M \star N )[ij] := \min_{k \in [n]}(M[{ik}]+N[{kj}])$.
\end{definition}

\begin{theorem}
\label{QMMtoAPSP}
If $\mm$ over $n \times n$ matrices with entries from~$[-n^c,n^c]$ can be solved by an $O(n^{2.5-\epsilon})$ time algorithm for some $\epsilon > 0$, then $\apsp$ over weighted graphs with $n$ nodes and entries from $[-n^c,n^c]$ can be solved in~$O(\log (n) \log(\log(n)) n^{2.5-\epsilon})$ time.
\end{theorem}

\begin{proof}
Let $G= (V,E)$ an $\apsp$ instance with weight matrix $W$ and with $n$ nodes. We apply the well-known (folkloric) repeated squaring method for computing the distance matrix described in detail in \cite{thesis_koen}. We use a $\mm$ subroutine to compute $W^{\star n}$ by making $\log(n)$ calls to the subroutine through repeated squaring. To boost the success probability of each of the $\log(n)$ subroutine calls we need to repeat them $O(\log(\log(n)))$ times.
\end{proof}

\begin{theorem}
\label{QAPSPtoMM}
If $\apsp$ for weighted graphs over $n$ nodes and with weights from $[-n^c,n^c]$ can be solved by an $O(n^{2.5-\epsilon})$ quantum algorithm for $\epsilon > 0$, then $\mm$ over $n \times n$ matrices and entries from~$[-n^c,n^c]$ can be solved in time $O(n^{2.5-\epsilon})$.
\end{theorem}

\begin{proof}
Let $M$ and $N$ be two $n\times n$ matrices. We will construct a graph in a way that we can read the product $M \star N$ from its distance matrix $D$. 

Let $G= (V,E)$ be the graph with vertex set $V = A \cup B \cup C$ such \mbox{that $A = \{a_1,,,a_n\}$}, $B= \{b_1,,,b_n\}$ and $C= \{c_1,,,c_n\}$. For every $M[ij]$ we add edge $(a_i,b_j)$ with weight $M[ij]$ to $E$ and for every $N[ij]$ we add edge~$(b_i, c_j)$ with weight $N[ij]$ to $E$. For a pair of nodes $a_i, c_j$ it holds that the distance of its shortest path is 
$\min_{b_k \in B} (w(a_i,b_k) + w(b_k, c_j)) = \min_{1\leq k \leq n} (M[ik] + N[kj]) = (M \star N)[ij]$.
\end{proof}
Again, we don't need any quantum operations for the above reduction to work. From \Cref{QMMtoAPSP} and \Cref{QAPSPtoMM} we now obtain the following:

\begin{corollary}
\label{QMMAPSPEQ}
For weighted graphs over $n$ nodes and with weights from~$[-n^c,n^c]$ there is a $\Tilde O(n^{2.5-\epsilon})$ time algorithm for $\apsp$ and $\epsilon > 0$ if and only if a $\Tilde O(n^{2.5-\epsilon})$ algorithm exists for $\mm$ over $n \times n$ matrices with entries from $[-n^c,n^c]$.
\end{corollary}

From $\mm$ we follow the classical chain of reductions, first proving a quantum conditional lower bound for $\apnt$, as was classically shown in \cite{williams_subcubic_2018}:

\begin{definition}[$\apnt$]
Given a tripartite weighted graph $G= (A\cup B \cup C, E)$ over $O(n)$ nodes and with weights from~$[-n^c,n^c]$ determine for every pair of nodes $a, b$ such that $a \in A$ and $b \in B$ whether there exists $c \in C$ such that nodes $a,b,c$ form a triangle of negative weight in $G$.
\end{definition}

\begin{theorem}
\label{QMMtoAPNT}
If $\apnt$ over weighted graphs of $n$ nodes and with weights from $[-n^c,n^c]$ can be solved by an $O(n^{2.5-\epsilon})$ time quantum algorithm for $\epsilon > 0$, then~$\mm$ over $n \times n$ matrices with entries from~$[-n^c,n^c]$ can be solved in $O(\log(n)\log(\log(n))n^{2.5-\epsilon})$ time. 
\end{theorem}

\begin{proof}
Given two $n \times n$ matrices $M$ and $N$ with entries from $[-n^c,n^c]$, we apply the reduction from \cite{williams_subcubic_2018}, constructing a graph of $O(n)$ nodes and perform a binary search over $[-n^c,n^c]$ to find the shortest paths using~$\log(n)$ applications of an $\apnt$ subroutine. To boost the individual $\apnt$ calls, we apply each one of them $\log(\log(n))$ times and take the output of each individual pair of nodes.  
\end{proof}


We find an interesting divergence in conditional lower bounds between our quantum and classical model in the reduction from $\apnt$ to $\nt$. 

\begin{definition}[$\nt$]
Given a weighted graph $G$ over $n$ nodes with weights from $[-n^c,n^c]$, determine if there is a triangle in $G$ with negative weight.
\end{definition}

\begin{theorem}
\label{QAPNTtoNT}
If $\nt$ can be solved in $O(n^{1.5-\epsilon})$ time by a quantum algorithm on weighted graphs over $n$ nodes with weights from $[-n^c,n^c]$ and for $\epsilon > 0$, then $\apnt$ on weighted graphs over $n$ nodes and with weights from $[-n^c,n^c]$ can be solved in $ O(\log(n)\log(\log(n))n^{2.5-\frac{\epsilon}{3}})$ time.
\end{theorem}

\begin{proof}
Given a weighted graph $G$ over $n$ nodes and with weights from $[-n^c,n^c]$, we apply the reduction from \cite{williams_subcubic_2018} to construct $n^{1-\alpha}$ graphs of size~$n^{\alpha}$ for $0\leq \alpha \leq 1$. If $\apnt$ can be solved in time $T(n)$, then we can solve $\nt$ in $O(T(n^{1-\alpha})(n^{3\alpha} + n^2))$ time. For $T(n) = n^{1.5-\epsilon}$, this is optimised at $\alpha=\frac23$. For each graph we need to make $\log(n)$ calls to the $\nt$ subroutine and to boost the success probability of each call, we apply each one of them $\log(\log(n))$ times. We can solve $\apnt$ in $O(\log(n)\log(\log(n))n^{2.5-\frac{\epsilon}{3}})$ time.
\end{proof}

Where in the classical setting we had cubic lower bounds for all problems from $\apsp$ to $\dmt$
and $\tc$, we find a $n^{1.5-o(1)}$ lower bound for $\nt$ for graphs with $n$ nodes. A lower bound of $n^{1.5-o(1)}$ does not imply that a better lower bound of e.g.~$n^{2.5-o(1)}$ does not exist for $\nt$. However, a simple Grover Search of all triples of nodes results in a matching upper bound, which \textit{does} imply that a larger lower bound would be inconsistent.

One could argue that it may come as more of a surprise that $\nt$ is as hard as $\apnt$ or $\apsp$ classically than that we find a gap in computational complexity in the quantum case. We look for a single triangle in $\nt$ and for $n^2$ potential triangles in $\apnt$. The quantum model highlights these gaps in difficulty of the problems in a way that the classical model could not. In this sense working in a quantum model on itself is already providing us with useful insights on the complexity of these problems, without requiring that we have a physical quantum algorithm to implement our reductions. An unfortunate consequence of this gap in lower bound complexities is that we lose the classical reduction from $\nt$ to $\mm$, since this reduction `as-is' does not let us go \textit{up} in complexity. In fact, the definition for fine-grained reductions makes it very challenging to prove computational lower bounds conditioned on smaller computational lower bounds, especially when the instances are similarly structured.

Less surprisingly, we find a similar lower bound of $n^{1.5-o(1)}$ for the $\zwt$ problem, using the reduction from \cite{vassilevska_finding_2009}. 

\begin{definition}[$\zwt$]
Given a weighted graph $G$ over $n$ nodes and with weights from $[-n^c,n^c]$, determine if there is a triangle in $G$ with weight 0.
\end{definition}

\begin{theorem}
\label{QNTto0WT}
If $\zwt$ on graphs over $n$ nodes and with weights from $[-n^c,n^c]$ can be solved in $O(n^{1.5-\epsilon})$ time by a quantum algorithm for some~$\epsilon > 0$, then $\nt$ on graphs of $n$ nodes and weights from~$[-n^c,n^c]$ can be solved in $O(\log(\log(n))n^{1.5-\epsilon})$ time.
\end{theorem}

\begin{proof}
Given a weighted directed graph $G$ over $n$ nodes with weights from $[-n^c,n^c]$, we apply the reduction from \cite{vassilevska_finding_2009} to construct $\log(n^c)$ graphs $G_i$ over $n$ nodes. Since we work in the adjacency matrix model of graphs, am $\nt$ algorithm has to look over $n^2$ matrix entries to decide whether a negative triangle exists in one of the graphs. This reduction therefore needs to be on-the-fly. The entries of the adjacency matrix of a graph $G_i$ can be computed in $O(1)$ time. Since we need to boost the individual calls to the $\zwt$ subroutine, we can detect a negative triangle using $O(\log(\log(n)))$ calls to a $\zwt$ algorithm. 
\end{proof}

Just like in the classical case, we find no reduction from $\zwt$ up the chain of reductions towards $\apsp$ or $\tsum$. Since $\zwt$ is the meeting point for $\tsum$ and $\apsp$, having a reduction from $\zwt$ to $\apsp$ or $\tsum$ would imply a reduction between $\tsum$ and $\apsp$.
Note that for all problems for which we have proven quantum lower bound up to this point, we can easily match this lower bound by a simple application of Grover search.

The reduction from $\zwt$ to $\dmt$ from \cite{abboud_matching_2018} highlights some interesting aspects about the complexity of $\dmt$.

\begin{theorem}
\label{Q0WTtoDMT}
If $\dmt$ on coloured graphs of $n$ nodes can be solved in $O(n^{1.5-\epsilon})$ time by a quantum algorithm for some $\epsilon > 0$ and $\omega(1) \leq \Delta(n) \leq o(\log(n))$, then $\zwt$ on weighted graphs of $n$ nodes and with weights from $[-n^c,n^c]$ can be solved in $O(\log(n)\log\log(n)n^{1.5-\epsilon})$ time by a quantum algorithm.
\end{theorem}

\begin{proof}
Given a weighted directed graph $G$ over $n$ nodes and using the reduction from the proof of \cite{abboud_matching_2018} we can compute the adjacency matrices of $2^{O(\Delta)}$ graphs $G_i$ over $O(\Delta n\cdot n^{\frac c\Delta})$ nodes on-the-fly. We apply a $\dmt$ algorithm to each of these graphs using on-the-fly computation. Computing the edges of our graphs requires functions computed using the lemma from \cite{abboud_matching_2018}. The functions are computable in $O(2^\Delta)$ time. It follows that for $\omega(1) \leq \Delta \leq o(\log(n))$ we can use a $\dmt$ algorithm to solve $\zwt$. We boost the $\log(n)$ calls to the $\dmt$ subroutine and arrive at a complexity of $O(\log(n)\log\log(n)n^{1.5})$.
\end{proof}

Additionally, we can use the following classical result from \cite{abboud_matching_2018} to prove quantum hardness of $\dmt$ for a larger range of $\Delta$ as mentioned in \Cref{q0wttodmt_poly}.

\begin{theorem}[\cite{abboud_matching_2018}]\label{polydelta}
If $\dmt$ on graphs with $n$ nodes can be solved by a $O(n^{3-\epsilon})$ time algorithm for some $\epsilon > 0$ and $\omega(1) \leq \Delta \leq n^{o(1)}$, then $\Delta'$-MT can be solved in $\Tilde O(n^{3-\epsilon})$ time for $\omega(1) \leq\Delta' \leq o(\log(n))$.
\end{theorem}

\begin{corollary}\label{q0wttodmt_poly}
If $\dmt$ on coloured graphs of $n$ nodes can be solved in $O(n^{1.5-\epsilon})$ time for some $\epsilon > 0$ and $\omega(1) \leq \Delta(n) \leq n^{o(1)}$, then $\zwt$ on weighted graphs of $n$ nodes and with weights from $[-n^c,n^c]$ can be solved in $\Tilde O(n^{1.5-\epsilon})$ time.
\end{corollary}

\begin{proof}
Follows from \Cref{Q0WTtoDMT} and \Cref{polydelta}.
\end{proof}

Here it is less clear whether the lower bound of $n^{1.5-o(1)}$ is the best lower bound we can find for $\dmt$. On an intuitive level, $\dmt$ definitely seems more complex than the single triangle finding problems of $\nt$ and $\zwt$. It is important to note that the above reduction only holds for the specified values of $\Delta$: $\omega(1) \leq \Delta(n) \leq n^{o(1)}$. We will see in Section~\ref{c5} that for these ranges of $\Delta$, there is indeed a matching upper bound. This leaves open the question of whether we can find a reduction from $\zwt$ to $\dmt$ for ranges of $\Delta$ that are polynomial or constant in the number of nodes in the graph. For polynomial values of $\Delta$ we are faced with the same challenge as in reducing $\nt$ to $\mm$ in the quantum case: we would be trying to increase the complexity of the lower bound through fine-grained reduction, going from $n^{1.5-o(1)}$ to potentially $n^{2.5-o(1)}$, depending on values of $\Delta$. In the next section we will see why $n^{2.5-o(1)}$ could be a reasonable quantum lower bound for $\dmt$ for unrestricted values of $\Delta$. 

On the other hand, the reduction from $\zwt$ to $\tc$ from \cite{abboud_matching_2018} makes use of the construction from the reduction from $\zwt$ to $\dmt$ for $\Delta$ values of $2^{O(\sqrt{\log n})}$, which is in the regime of $\Delta$ where we found an $n^{1.5-o(1)}$ lower bound for $\dmt$. As a consequence, we get a similar result for $\tc$.

\begin{theorem}
\label{Q0WTtoTC}
If $\tc$ can be solved in $O(n^{1.5-\epsilon})$ time by a quantum algorithm for some $\epsilon > 0$, then $\zwt$ can be solved in $O(\log(n)\log\log(n)n^{1.5-\epsilon})$ time by a quantum algorithm.
\end{theorem}

\begin{proof}
Given a weighted directed graph $G$ over $n$ nodes, using the reduction from the proof of \cite{abboud_matching_2018} we can compute the adjacency matrices of $2^{O(\Delta)}$ graphs $G_i$ over $O(\Delta n\cdot n^{\frac c\Delta})$ nodes on-the-fly. We set $\Delta = 2^{\sqrt{\log(n)}}$ and use a $\tc$ algorithm to solve $\zwt$, boosting each of the individual $\tc$ calls.
\end{proof}

Here the question now really comes down to whether we can find a $O(n^{1.5})$ matching upper bound for $\tc$, which we do in the Section~\ref{c5}.

\subsection{Quantum Fine-Grained Reductions from $\cnfsat$ and $\ksat$}

Through the results from the previous section, and the quantum reductions from $\tsum$ to $\zwt$ from \cite{buhrman_limits_2021} we now have lower bounds on $\dmt$ and $\tc$ from the disjunction of hardness conjectures for $\apsp$ and $\tsum$. To complete the result, we need to verify that the reductions from $\ksat$, both taken from \cite{abboud_matching_2018}, still hold in the quantum case.  For $\ksat$ we go through the Sparsification Lemma from \cite{impagliazzo_complexity_2001}, arriving at the same lower bound we found through $\zwt$ for the same ranges of $\dmt$.

\begin{lemma}[Sparsification Lemma, \cite{impagliazzo_complexity_2001}]\label{sparse}
Let $\phi$ be a $k$-CNF formula over $n$ variables and with $m$ clauses for $k \geq 3$. For any $\epsilon > 0$ there is an $O(2^{\epsilon n})$ time algorithm that produces $O(2^{\epsilon n})$ $k$-CNF formulas $\phi_1,\dots,\phi_{O(2^{\epsilon n})}$ over $n$ variables and $c n$ clauses where $c = (\frac{k}{\epsilon})^{O(k)}$. It then holds that $\phi$ is in $\sat$ if and only if there is a satisfying assignment to $\bigvee_{i=1}^{O(2^{\epsilon n})}\phi_i$. 
\end{lemma}

\begin{theorem}
\label{QkSATtoDMT}
If $\dmt$ on a graph of $N$ nodes can be solved in $O(N^{1.5-\epsilon})$ time by a quantum algorithm for some $\epsilon > 0$ and $\omega(1) \leq \Delta(N) \leq N^{O(1)}$, then $\ksat$ on $n$ input variables can be solved in $\Tilde O(2^{\frac n2 (1-\epsilon')})$ time for some $0<\epsilon'<\frac23 \epsilon$.
\end{theorem}

\begin{proof}
Suppose we have a $O(N^{\frac 32-\epsilon})$ time $\dmt$ algorithm for graphs with $N$ nodes and $\omega(1) \leq \Delta(N) \leq N^{o(1)}$ and let $\phi$ be a $k$-CNF instance over $n$ variables and $m$ clauses. We apply the sparsification lemma, \Cref{sparse}, to compute $2^{\epsilon' n}$ $k$-CNF formulas $\phi_i$ with $cn$ clauses for some constant value $c$ in $2^{\epsilon' n}$ time. We then apply the reduction from \cite{abboud_matching_2018} to decide whether $\phi_i$ is satisfiable for each $i \in [2^{\epsilon' n}]$ in $O((\Delta 2^{\frac n3 + \frac{nc}{3\Delta}})^{\frac32 -\epsilon})$ time. Since $\omega(1) \leq \Delta(N) \leq N^{o(1)}$ it follows that $O((\Delta 2^{\frac n3 + \frac{nc}{3\Delta}})^{\frac32 -\epsilon}) = O(2^{n(\frac12 - \frac\epsilon3) + o(1)})$. The total time to evaluate all sparse formulas will be $O(2^{n(\frac12 - \frac\epsilon3 + \epsilon') + o(1)}) = O(2^{\frac n2(1 - \frac{2}{3}\epsilon + 2\epsilon') +o(1) })$. The theorem statement follows for $0<\epsilon' < \frac\epsilon6$. For these ranges of $\epsilon'$, the number of sparse formulas and the time required to compute them will not exceed the time necessary to solve $\ksat$ faster than conjectured. 
\end{proof}

For any $\epsilon' < \frac12$, the sparsification lemma can easily be applied for quantum reductions from $\ksat$. Furthermore, we don't need to explicitly construct all graphs for our sparse formulas as we did in the proof above. Instead, we can Grover search the set of sparse formulas, applying our $\dmt$ reduction on-the-fly. In this case, the total run time can be improved to $O(2^{\frac n2 (1 -\frac{2\epsilon}{3} + \epsilon')})$, allowing for a wider range of $\epsilon' < \frac\epsilon3$. The reduction from $\sat$ to $\tc$ gives us the same lower bound as we found from $\zwt$. 

\begin{theorem}
\label{QSATtoTC}
If $\tc$ on a graph of $N$ nodes can be solved in $O(N^{\frac32-\epsilon})$ time by a quantum algorithm for some $\epsilon > 0$, then $\sat$ on $n$ input variables can be solved in $\Tilde O(2^{\frac n2 (1-\frac\epsilon3)})$ time.
\end{theorem}

\begin{proof}
Suppose we have a $O(N^{1.5 -\epsilon})$ time algorithm for $\tc$ on graphs of $N$ nodes. Let $\phi$ be a CNF formula over $n$ variables and $m$ clauses. We apply the reduction from \cite{abboud_matching_2018} to produce a graph over $O(2^{\frac n3}m)$ nodes. We can then use our $\tc$ algorithm to solve $\sat$ in $\Tilde O(2^{\frac{n}{2}(1-\frac\epsilon3)})$ time.
\end{proof}

Since we find the same lower bounds from $\zwt$ and $k$ $\sat$ to $\dmt$ and $\tc$ for the same ranges of $\Delta$, we can now state the following.

\begin{theorem}\label{Qextremelyweakconjecture}
If $\dmt$ with $\omega(1) \leq \Delta \leq n^{o(1)}$ or $\tc$ can be solved in time $O(n^{1.5-\epsilon})$ for some $\epsilon > 0$, then Conjecture~\ref{conj:PopularConjecture} must be false. 
\end{theorem}

\begin{proof}
We know that $k$-SAT reduces to $\dmt$ for $\omega(1) \leq \Delta \leq n^{o(1)}$ from \Cref{QkSATtoDMT}  and reduces to $\tc$ from \Cref{QSATtoTC}. $\apsp$ reduces to $\zwt$ from \Cref{QAPSPtoMM}, \Cref{QMMtoAPNT}, \Cref{QAPNTtoNT} and \Cref{QNTto0WT} and $\tsum$ reduces to $\zwt$ as shown in \cite{buhrman_limits_2021}. We then found that $\zwt$ reduces to $\dmt$ for $\omega(1) \leq \Delta \leq n^{o(1)}$ from \Cref{q0wttodmt_poly} and to $\tc$ from \Cref{Q0WTtoTC}. 
\end{proof}

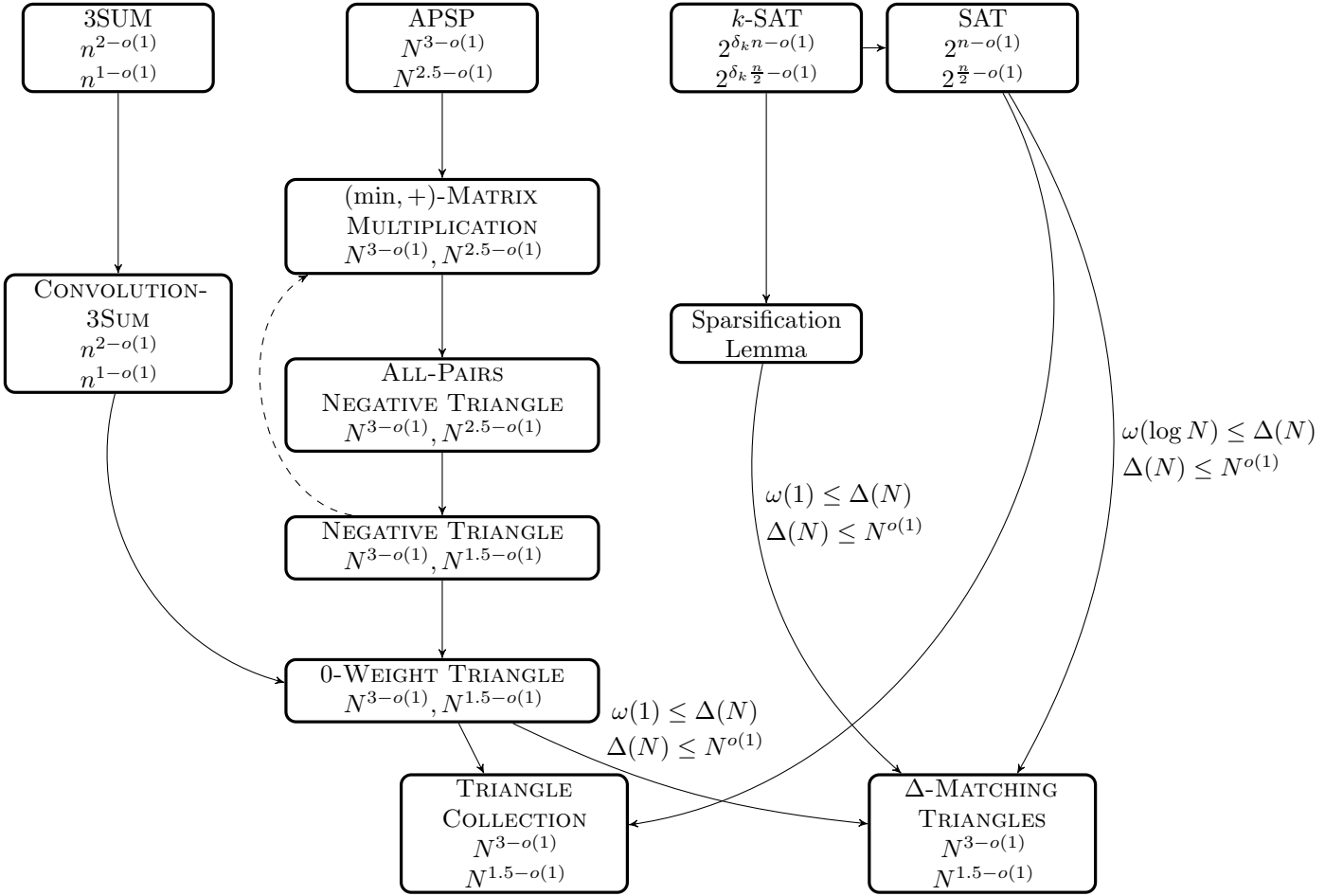
\begin{figure}[h]
\hspace*{-1.5cm}
    \label{fig:my_label}
\begin{tikzpicture}[->,>=stealth']

 \node[state,
 text width=2.5cm] (3SUM)
 {3SUM\\
 $n^{2-o(1)}$\\$n^{1-o(1)}$};
 
 \node[state,
 below of=3SUM,
 text width=3cm,
 node distance=4cm] (CONV)
 {\textsc{Convolution-}\\ \textsc{3Sum}\\
 $n^{2-o(1)}$\\$n^{1-o(1)}$};
 
 \node[state,
 node distance=4.5cm,
 right of=3SUM,
 text width=2.5cm] (APSP) 
 {APSP\\
 $N^{3-o(1)}$\\$N^{2.5-o(1)}$};
 
 \node[state,
 below of=APSP,
 node distance=2.5cm,
 text width=4.2cm] (min+)
 {\textsc{$(\min,+)$-Matrix} \\\textsc{Multiplication}\\
 $N^{3-o(1)},N^{2.5-o(1)}$};
 
 \node[state,
 below of=min+,
 text width=4.2cm,
 node distance=2.5cm] (APNT)
 {\textsc{All-Pairs}\\\textsc{Negative Triangle}\\
 $N^{3-o(1)},N^{2.5-o(1)}$};
 
 \node[state,
 below of=APNT,
 text width=4.2cm,
 node distance=2cm] (NT)
 {\textsc{Negative Triangle}\\
 $N^{3-o(1)},N^{1.5-o(1)}$};

 \node[state,
 below of=NT,
 text width=4.2cm,
 node distance=2cm] (0-WT)
 {$\zwt$\\
 $N^{3-o(1)},N^{1.5-o(1)}$};
 
 \node[state,
 below of=0-WT,
 text width=3cm,
 node distance=2cm,
 xshift=7.5cm] (delta)
 {\textsc{$\Delta$-Matching} \\\textsc{Triangles}\\
 $N^{3-o(1)}$\\$N^{1.5-o(1)}$};
 
 \node[state,
 left of=delta,
 text width=3cm,
 node distance=6.5cm] (TC)
 {\textsc{Triangle}\\\textsc{Collection}\\
 $N^{3-o(1)}$\\$N^{1.5-o(1)}$};

 \node[state,
 node distance=4.5cm,
 right of=APSP,
 text width=2.5cm] (kSAT) 
 {$k$-SAT\\
 $2^{\delta_k n - o(1)}$\\$2^{\delta_k \frac n2 - o(1)}$};
 
 \node[state,
 below of=kSAT,
 text width=2.5cm,
 node distance=4cm] (sparse)
 {Sparsification \\ Lemma};
 
 \node[state,
 node distance=3cm,
 right of=kSAT,
 text width=2.5cm] (SAT) 
 {SAT\\
 $2^{n-o(1)}$\\$2^{\frac n2 - o(1)}$};

     \path  (APSP)  edge    node[anchor=south,below]{}
                    node[anchor=north,above]{} (min+)
            (min+)  edge    node[anchor=south,below]{}
                    node[anchor=north,above]{} (APNT)
            (APNT)  edge    node[anchor=south,below]{}
                    node[anchor=north,above]{} (NT)
            (NT)    edge    node[anchor=south,below]{}
                    node[anchor=north,above]{} (0-WT)
            (NT)    edge[bend left=70,dashed]  
                        node[anchor=west,left]{}
                        node[anchor=west,left]{} (min+)
            (0-WT)  edge[bend right=10]    node[anchor=south,above,xshift=0cm,yshift=.8cm]{$\omega(1) \leq \Delta(N)$}
                            node[anchor=north,above,xshift=-0cm,yshift=.3cm]{$\Delta(N) \leq N^{o(1)}$} (delta)
            (0-WT)  edge    node[anchor=south,below]{}
                            node[anchor=north,above]{} (TC)
            (3SUM)  edge    node[anchor=south,below]{}
                            node[anchor=north,above]{} (CONV)
            (CONV)  edge[bend right=45]    node[anchor=south,below]{}
                            node[anchor=west,left]{} (0-WT)
            (kSAT)  edge    node[anchor=south,below]{}
                            node[anchor=north,above]{} (sparse)
            (sparse)  edge[bend right=30]    node[anchor=south,right,yshift=.8cm,xshift=-.2cm]{$\Delta(N)\leq N^{o(1)}$}
                            node[anchor=north,right,yshift=1.3cm,xshift=-.2cm]{$\omega(1) \leq \Delta(N)$ } (delta)
            (SAT)  edge[bend left]    node[anchor=east,right]{$\omega(\log N) \leq \Delta(N)$ }                node[anchor=east,right,yshift=-.5cm]{$\Delta(N) \leq N^{o(1)}$} (delta)
            (SAT)  edge[bend left=55]    node[anchor=south,below]{}
                            node[anchor=south,below]{} (TC)
            (kSAT)  edge    node[anchor=west,below]{}
                            node[anchor=east,above]{} (SAT)
                            ;
\end{tikzpicture}
    \caption{Quantum fine-grained reductions to $\dmt$ and $\tc$. The parameter $n$ denotes the size of the sets in the case of $\tsum$ and CONV-3SUM and the number of variables in the case of $\sat$ and $\ksat$. The parameter $N$ denotes the number of nodes in a graph. The first lower bound in each node denotes the known classical lower bound and second lower bound denotes the new quantum lower bound. In the case where multiple edges arrive at one node, the lower bound in the second line of the node is the same for all reductions. For the dashed edge we only know of a classical reduction. Lastly, the lower bounds for $\dmt$ hold only for the $\Delta$ values denoted by the labels on the incoming edges.}
        \label{fig:flowchart}
\end{figure}
\clearpage

\section{$\dmt$ and $\tc$: Upper Bounds}
\label{c5}
In the previous section we saw a series of lower bounds, conditioned on our quantum hardness conjectures. For most of these problems we find matching upper bounds that follow from trivial applications of Grover Search, except for $\dmt$ and $\tc$, for which we find upper bounds by clever use of data structures and Ambainis' Variable Time Grover Search. 
\begin{theorem}[Variable Time Grover Search (VTGS) \cite{ambainis_quantum_2006}]\label{TVGS}
Let $x \in \{0,1\}^n$ and suppose we can compute the value of $x_i$ in time $t_i$.  There exists a quantum algorithm such that the following hold: If there is some $i \in \{0,1\}^{\log n}$ such that $x_i=1$, then the algorithm will output some $j \in \{0,1\}^{\log n}$ such that $x_j=1$ with probability~$\frac23$. If there is no such $i \in \{0,1\}^{\log n}$, then the algorithm will output $0$ with probability~$\frac23$. The algorithm takes $O(\sqrt{\sum_{i=1}^n t_i^2})$ time and makes $O(\sqrt{\sum_{i=1}^n t_i^2})$ queries to~$x$.
\end{theorem}

\subsection{Quantum Algorithm for $\dmt$}
Recall that for $\dmt$, given a graph as an input we want to find out whether there exists a colour triple for which there are at least $\Delta$ triangles. Depending on the value of $\Delta=n^{\alpha}$, we exhibit an algorithm that invokes one of the following two subroutines: For small $\alpha$, it uses brute force with a variant of Grover Search, while for large $\alpha$, our subroutine uses matrix multiplication. Additionally, for the first case our algorithm requires fast access to the list of vertices coloured by $i$ for every colour $i\in \Gamma$. For that we use a data structure, which with the help of quantum random-access gates $\RAG$s (defined in Equation~\ref{eq:RAG}), allows us to index in $O(1)$ time \footnote{The constant time is in the word-RAM model.}. On the other hand, for large $\alpha$, notice that as $\alpha$ increases, the amount of colour triples for which there are even enough nodes in the graph to form $n^\alpha$ triangles decreases because the number of triples of nodes any graph $G$ (of $n$ nodes) can have is $n^3$. Therefore, number of colour triples $(i,j,k) \in \Gamma^3$ that can have at least $\Delta$ triangles, must have at least $\Delta$ different triples of vertices, is at most $\frac{n^3}{\Delta}$. Moreover, as we have computed the values of $|V_i|$ for every $i \in \Gamma$ in the pre-processing step we can easily filter out the un-promising candidates and only check if there exists $\Delta$ triangles (using matrix multiplication) for the colour triples with enough nodes of that colour triple. 

\begin{theorem}\label{QDMTu}
There exists a quantum algorithm that solves $\dmt$ on a graph of $n$ nodes in $\tO(\min\{n^{1.5 + \frac\alpha2},n^{1.5 + \omega -\frac\alpha2}\})$ time with $\Delta = n^\alpha$ for $0\leq \alpha \leq 3$.
\end{theorem}
\begin{proof}
Let $G=(V,E)$ be the input, a coloured graph with $|V|=n$ and colours given by $\gamma: V \rightarrow \Gamma$. Let $\Delta=n^\alpha$ for a given $\alpha \in [0,3]$. Given a colour $i \in \Gamma$, let $V_i\subseteq V$ denote the subset containing only $i$-coloured nodes, and use $|V_i|$ to denote the number of vertices in $V_i$. 

\paragraph{Pre-processing step.} In $O(n)$ time one can compute values $|V_i|$ for all $i \in \Gamma$. Having computed the $|V_i|$ for all $i \in \Gamma$, create a hash-table of $|\Gamma|$ buckets. Each bucket is indexed by $i \in |\Gamma|$ containing an array of size $|V_i|$, respectively. We go over all the vertices $v\in V$ and place them in the hash table corresponding to the bucket indexed by $\gamma(v)$, i.e., the colour of vertex $v$. This entire process takes $O(n)$ time and uses $O(n)$ space.

\paragraph{Subroutine for small $\alpha$.} We use VTGS on the set $\Gamma^3$. For every colour triple $(i,j,k) \in \Gamma^3$, let $t_{i,j,k}$ denote the time taken to determine whether $G$ contains $\Delta$ triangles of colour triple $(i,j,k)$. This can be done in $t_{i,j,k} = \tO(\sqrt{\Delta |V_i||V_j||V_k|})= \tO(\sqrt{n^\alpha|V_i||V_j||V_k|})$ time using (a variant of) Grover Search where the queries to this Grover-like subroutine are indexed by $(a,b,c) \in [|V_i|]\times[|V_j|]\times[|V_k|]$. Note that, the $\tO(\cdot)$ hides the $\poly(\log(n))$ factors that arise because we want the probability of failure to be reduced to $\frac{1}{\poly(n)}$. Moreover, every such query to the Grover-like subroutine can be implemented in three queries to the adjacency matrix given to us as input, in the following way: Access the $a^{\text{th}}$ node of hash bucket containing $V_i$, $b^{\text{th}}$ node of hash bucket, which is the data structure that we used to pre-process the nodes, containing $V_j$, $c^{\text{th}}$ node of hash bucket containing $V_k$, and, check if there is a triangle labelled by nodes $a,b,c$ in $G$.

Since we can use VTGS, the total time taken for this algorithm is $T(n) = O(\sqrt{\sum_{i,j,k \in \Gamma}t^2_{i,j,k}})$, and moreover, because we have $\sum_{i \in \Gamma}|V_i| = n$ we find the following.
\begin{align*}
       T(n) &= O(\sqrt{\sum_{i,j,k \in \Gamma}t^2_{i,j,k}}) 
       =\tO(\sqrt{\sum_{i,j,k \in \Gamma} n^{\alpha}|V_i| |V_j| |V_k|})\\
       &=\tO(\sqrt{n^{\alpha}\sum_{i \in \Gamma}|V_i|\sum_{j \in \Gamma}|V_j|\sum_{k \in \Gamma}|V_k|}) 
       = \tO(\sqrt{n^{\alpha}n^3}) = O(n^{1.5+\frac\alpha2}).
\end{align*} 
\paragraph{Subroutine for large $\alpha$.} For every triple $(i,j,k) \in \Gamma^3$, let $t_{i,j,k}$ denote the time taken to compute the following:
\begin{itemize}
    \item Step 1: Check if $|V_i|\cdot |V_j| \cdot |V_k| < \Delta$. If yes then do nothing more, in that case $t_{i,j,k}=O(1)$. Else, check Step~2 below.
    \item Step 2: Check if there are at least $\Delta$ triangles of colour triple $(i,j,k)$ in $G$. As we are in the large $\alpha$ regime we do not want to use threshold version of Grover search, instead we use matrix multiplication to compute the cube of the adjacency matrix restricted to the entries only coloured by coloured by $i,j,k$ and taking trace of the diagonal. A \emph{yes} instance is when the computed trace value is greater than equal to $\Delta$, else it is a \emph{no} instance. The total time taken in Step~2 is $t_{i,j,k}=\tO(n^{\omega})$. Here the $\tO(\cdot)$ hides the $\poly(\log(n))$ factors that arise because we want the probability of failure to be $\frac{1}{\poly(n)}$.\footnote{With this, using union bound we can argue that the total error probability of the main algorithm is at most a small constant.}
\end{itemize}
Now to analyse the time taken for the large $\alpha$ case: Using VTGS we again have
\begin{align*}
    T(n) &= O(\sqrt{\sum_{i,j,k \in \Gamma}t^2_{i,j,k}})
    = O(\sqrt{\sum_{\substack{i,j,k \in \Gamma\\\text{s.t. }|V_i||V_j||V_k|\ge \Delta}} \hspace{-2em} t^2_{i,j,k} + \sum_{\substack{i,j,k \in \Gamma\\\text{s.t. }|V_i||V_j||V_k|< \Delta}} \hspace{-2em} t^2_{i,j,k}})\\
    &=\tO(\sqrt{\sum_{\substack{i,j,k \in \Gamma\\\text{s.t.\ }|V_i||V_j||V_k|\ge \Delta}} \hspace{-2em} n^{2\omega} + \sum_{\substack{i,j,k \in \Gamma\\\text{s.t.\ }|V_i||V_j||V_k|< \Delta}} \hspace{-2em} O(1)}) 
    \le \tO(\sqrt{\frac{n^{2\omega} \cdot n^3}{\Delta} + n^3})= \tO(\sqrt{\frac{n^{2\omega} \cdot n^3}{\Delta}})=\tO(n^{\omega+1.5-\frac{\alpha}{2}}).
\end{align*}
\paragraph{Combined approach.} Given an input instance of $\dmt$, we first compute which of the two approaches is faster based on $\Delta=n^\alpha$ and then apply that approach. We find $T(n) = \tO(\min\{n^{1.5 + \frac\alpha2},n^{1.5 + \omega -\frac\alpha2}\})$.

Setting 
$1.5 + \frac\alpha2 = 1.5 + \omega - \frac\alpha2$,
we see that for $\alpha<\omega$ the small $\alpha$ algorithm will be faster while for $\alpha > \omega$ our large $\alpha$ algorithm is faster. Therefore, $\dmt$ for any $\Delta=n^{\alpha}$ can be solved in
$\tO(\min\{n^{1.5 + \frac\alpha2},n^{1.5 + \omega -\frac\alpha2}\})$ time quantumly. 
\end{proof}

From \Cref{QDMTu} we get a worst case corollary and a corollary for the range of $\Delta$ for which we found reductions in \Cref{c4}. 

\begin{corollary}
There exists a quantum algorithm that solves $\dmt$ on graphs of $n$ nodes in $\tO(n^{1.5 + \frac\omega2})$ time for any $\Delta$.
\end{corollary}

Note that this means that for the current value of $\omega \approx 2.3728$, the matrix-multiplication constant, we can solve $\dmt$ sub-cubically on a quantum computer for any range of $\Delta$.

\begin{corollary}\label{qDMTsmallD}
There exists a quantum algorithm that solves $\dmt$ on graphs of $n$ nodes in $\tO(n^{1.5+o(1)})$ time for $\Delta \leq n^{o(1)}$.
\end{corollary}

The commonly conjectured lower bound for $\omega$ is $2$, in which case we have a worst case complexity for $\dmt$ of $\tO(n^{2.5})$, matching the quantum complexity of other problems encountered in our reductions from $\apsp$. That is not to say that a faster algorithm is not possible of course. Pushing the $n^{1.5-o(1)}$ lower bound for $\dmt$ up for polynomial $\Delta$ is challenging with current techniques and a matching upper bound of $\tO(n^{1.5})$ for unrestricted $\Delta$ is not impossible, albeit unlikely. 

\subsection{Quantum Algorithm for $\tc$}
Using a similar strategy as for the case of small $\Delta$ in the proof of \Cref{QDMTu} we also find a tight upper bound for $\tc$. Recall that for $\tc$ we want to know whether there is a triangle for every possible colour triple in a given graph. Equivalently, we may ask whether there exists a colour triple in a graph for which there is no triangle.

\begin{theorem}\label{QTCu}
There exists a quantum algorithm that solves $\tc$ in $\tO(n^{1.5})$ time.
\end{theorem}

\begin{proof}
Let $G=(V,E)$ be a coloured graph with $|V|=n$ and colours given by $\gamma: V \rightarrow \Gamma$. We use the Variable-Time Grover Search algorithm to the set of all colour triples $\Gamma^3$, to determine whether there is a colour triple for which there is no triangle in $G$. Let $(i,j,k)$ be a triple of colours and let the time it takes to check whether $G$ contains a triangle in these colours be $t_{i,j,k}$. Then the algorithm takes $T(n) = O(\sqrt{\sum_{i,j,k \in \Gamma}t^2_{i,j,k}})$ time. Given a colour $i$, let $V_i\subseteq V$ be the subset containing only nodes in that colour and let $|V_i|$ be the number of vertices in $V_i$. With some pre-processing of the input, just like it is done in the proof of Theorem~\ref{QDMTu}, for any triple of colours $(i,j,k)$, we can Grover Search over the set $V_i\times V_j \times V_k$ to determine whether there is a triangle in the induced sub-graph in $\tO(\sqrt{|V_i| \cdot |V_j|\cdot |V_k|})$ time, and we have that $t^2_{i,j,k} = \tO(|V_i| \cdot |V_j|\cdot |V_k|)$. Since it holds that $\sum_{i\in \Gamma}|V_i| = n$ we have
\begin{align*}
       T(n) &= O(\sqrt{\sum_{i,j,k \in \Gamma}t^2_{i,j,k}}) 
       =\tO(\sqrt{\sum_{i,j,k \in \Gamma} |V_i| \cdot |V_j|\cdot |V_k|})
       =\tO(\sqrt{\sum_{i \in \Gamma}|V_i|\sum_{j \in \Gamma}|V_j|\sum_{k \in \Gamma}|V_k|}) 
       =\tO(\sqrt{n^3}) = \tO(n^{1.5}).
\end{align*}
Therefore, we can solve $\tc$ in $\tO(n^{1.5})$ time quantumly. 
\end{proof}

\section{Discussion and Future Work}
\label{c6}
In the course of proving quantum time lower bounds for $\dmt$ and $\tc$, we also prove conditional lower bounds for many other related problems in the process and were able to make some interesting observations. In the classical setting, it can be seen as surprising that seemingly simple problems such as $\nt$ and $\zwt$ are at least as hard as problems like $\apsp$ and $\mm$, which seem more complex. With our results in the quantum setting we find that this intuition regarding the difference in complexity of these problems (to a reasonable degree) is well-founded, because here a quadratic gap in complexity in both the lower and the upper bounds of these problems is witnessed instead. 

This paper leaves us with numerous future directions. For example,
\begin{itemize}
    \item In the classical case, it is possible to find a fine-grained reduction from $\nt$ to $\mm$, but it seems much harder to find an equivalent reduction in the quantum case.
    To pick up where this work has left off would be to face the challenges of finding such a reduction.
    \item The $\dmt$ problem presents another series of worthwhile challenges. Due to the dependency of the upper bound on the matrix-multiplication exponent~$\omega$, faster algorithms for $\dmt$ can definitely be found, either by improving $\omega$, or by finding an algorithm with no or lower dependency on $\omega$. The question remains then whether there exists quantum algorithms that have a better than $\tO(n^{2.5})$ worst case $\Delta$ complexity, which would be the worst case $\Delta$ complexity if $\omega=2$. To reinforce the likelihood of $\tO(n^{2.5})$ being the best worst case $\Delta$ upper bound for \textsc{$\Delta$-Matching Triangles}, we could look for reductions to \textsc{$\Delta$-Matching Triangles} that work for polynomial ranges of $\Delta$. In the classical setting, we saw in \cite{abboud_matching_2018} that for constant ranges of $\Delta$, \textsc{$\Delta$-Matching Triangles} permits a faster than cubic algorithm, and it leads us to wonder whether a faster than $\tO(n^{1.5})$ algorithm exists for \textsc{$\Delta$-Matching Triangles} with constant $\Delta$ in the quantum case. The \textsc{$\Delta$-Matching Triangles} and \textsc{Triangle Collection} problems are of course of special interest due to their connection to all three popular hardness conjectures.
    \item Much work remains to be done in quantum fine-grained complexity as a whole. A natural starting point is of course the translation of classical reductions to the quantum case (for which a good starting point can be, e.g., \cite{williams_fine-grained_2019}), but in the long term it is worthwhile to consider which reductions are only possible in the quantum case -- to find a true quantum complexity web of fine-grained reductions, guiding the possibilities and impossibilites of quantum algorithm design.
\end{itemize}

\section{Acknowledgements}
Subhasree Patro is supported by the Robert Bosch Stiftung. Harry Buhrman, Subhasree Patro, and
Florian Speelman are additionally supported by NWO Gravitation grants NETWORKS and QSC, and
EU grant QuantAlgo. This
work was supported by the Dutch Ministry of Economic Affairs and Climate Policy (EZK), as part of
the Quantum Delta NL programme.

\bibliographystyle{alpha}
\bibliography{references}
\end{document}